\documentclass[12pt]{article}
\usepackage[margin=1in]{geometry}  
\usepackage[singlespacing]{setspace}
\usepackage{epsfig}
\usepackage{graphics,graphicx,color}
\usepackage{amssymb,amsfonts,amsthm,amsmath}
\usepackage{authblk}
\usepackage{multirow}
\usepackage{mathrsfs}
\usepackage{enumerate}
\usepackage{natbib}
\usepackage{float}
\usepackage{algorithm}
\usepackage{algpseudocode}
\usepackage{subcaption} 
\usepackage{url}
\usepackage{bm}
\usepackage{comment}
\usepackage{siunitx}
\usepackage{xcolor}

\usepackage{lscape}
\usepackage{mathtools}
\usepackage{stackengine} 
\usepackage{indentfirst}
 
\usepackage{array,arydshln}

\newtheorem{definition}{Definition}
\newtheorem{Lemma}{Lemma}
\newtheorem{example}{Example} 
\newtheorem{proposition}{Proposition}

\newtheorem{remark}{Remark}
\newtheorem{const}{Construction}

\newcommand{\bA}{\textbf{A}}
\newcommand{\bB}{\textbf{B}}
\newcommand{\bD}{\textbf{D}}

\newcommand{\bE}{\textbf{E}}

\newcommand{\ba}{\textbf{a}}
\newcommand{\bb}{\textbf{b}}
\newcommand{\bd}{\textbf{d}}
\newcommand{\bc}{\textbf{c}}

\newcommand{\bff}{\textbf{f}}

\newcommand{\bzero}{\textbf{0}}

\newcommand{\be}{\textbf{e}}

\usepackage{threeparttable}

\newcommand{\red}{\color{red}}

\newcommand\oast{\stackMath\mathbin{\stackinset{c}{0ex}{c}{0ex}{\ast}{\bigcirc}}}

\title{Grouped Orthogonal Arrays from Orthogonal Arrays and Difference Schemes} 

\author[1]{Meixin Liu}
\author[2]{Chunyan Wang$^*$}
\author[3]{Guanzhou Chen}
\author[1]{C. Devon Lin$^*$}%

\affil[1]{Department of Mathematics and Statistics, Queen's University}
\affil[2]{Center for Applied Statistics and School of Statistics, Renmin University of China}
\affil[3]
{NITFID, LPMC, KLMDASR, School of Statistics and Data Science, Nankai University}
 
\begin{document}

\date{}
\maketitle
\begin{abstract}
Grouped orthogonal arrays were introduced to address experimental design problems arising in computer experiments with grouped inputs, as well as in physical experiments where interactions between factors from different groups are assumed to be negligible. Motivated by the growing need for flexible and efficient designs under such settings, this article develops several  constructions to expand the existing catalogs of grouped orthogonal arrays. The proposed constructions provide a large collection of new grouped orthogonal arrays with significantly larger numbers of groups and group sizes.   

\end{abstract}

\section{Introduction}
 
Design of experiments is a fundamental area in statistics and data science with a long and rich history. Experimental designs play a crucial role in supporting effective decision-making in science and engineering   \cite{wu2011experiments,gramacy2020surrogates,lin2025orthogonal,lin2025data,zhou2026fast}. 
The choice of a design typically depends on the underlying data structure and the available prior information about the factors. 
The three fundamental effect principles (sparsity, heredity, hierarchy) can be viewed as an example of \emph{prior} knowledge. 

One emerging direction in the design of experiments is the incorporation of  priori knowledge regarding the relationships among input factors and response variables. In particular, \cite{10.1093/biomet/ass035} motivated this line of research through an engineering application involving a two-step surface-mounting process. In this application, interactions between factors from the two steps were assumed to be negligible, whereas interactions among factors within the same step could be significant. Based on this motivation, \cite{10.1093/biomet/ass035} introduced the concept of \emph{designs of variable resolution} for experiments in which factors are partitioned into disjoint groups, under the assumption that interactions do not occur between factors from different groups.

More recently, \cite{chen2025} and \cite{li2025grouped} extended this line of work from physical experiments to computer experiments, digital twins, and related applications through the development of \emph{grouped orthogonal arrays}. Both works numerically demonstrated that such arrays possess superior space-filling properties and achieve higher prediction accuracy compared with designs obtained by randomly permuting the columns of the corresponding orthogonal arrays.  \cite{chen2025} also argued that grouped orthogonal arrays may be particularly useful in models with blocked (or grouped) additive kernels. 

Several constructions for grouped orthogonal arrays have been proposed. \cite{10.1093/biomet/ass035} and \cite{lekivetz2016designs} developed constructions for two-level designs, while \cite{raaphorst2014construction} considered designs with at most three factors per group and run size $s^3$ for a prime power $s$. The constructions in \cite{zhang2023construction} require run sizes of $s^t$ for an integer $t \geq 4$ and allow only one group with higher strength. More recently,  \cite{chen2025} extended the study to grouped orthogonal arrays with higher within-group strengths and developed several construction methods, most of which are based on regular designs and thus yield primarily prime-power run sizes.  \cite{li2025grouped}  proposed direct and recursive constructions with more flexible run sizes; however, their designs are limited to the case where each group has strength three.


We propose several constructions for grouped orthogonal arrays to expand existing design catalogs. Unlike existing methods, our approaches combine two orthogonal arrays to produce grouped orthogonal arrays with substantially larger run sizes and more groups. Another key contribution is the use of difference schemes of strength $t$ \cite{HEDAYAT1996307}, which, despite being introduced decades ago, have seen limited application in the literature and practice.  Our work shows that they provide an effective tool for constructing grouped orthogonal arrays. In many instances, the proposed constructions produce new grouped orthogonal arrays with larger group sizes  that cannot be obtained through existing methods. 


The remainder of the article is organized as follows. Section 2 gives some notation and definitions.
Sections 3 and 4 provide the constructions using orthogonal arrays and difference schemes of strength $t$, respectively. Section 5  concludes the
article.
  
 \section{Definitions, Notation and Background}\label{sec2}
In this section, we briefly review some necessary definitions, notation, and background in Galois field, difference schemes and some relevant constructions of orthogonal arrays. For more in-depth explanations of these concepts, see \cite{hedayat1999orthogonal} and \cite{chengtang2025}.

Throughout this article, we focus on orthogonal arrays in which each factor has $s$ levels, where $s$ is a  prime power.  Let $\mathbf{D}$ be a design with $N$ runs and $k$ factors of $s$ levels. If for every $N \times t$ submatrix of 
$\mathbf{D}$,  each level  combination appears equally often, then $\mathbf{D}$ is called an {\em orthogonal array} of strength $t$. We use $OA(N,k,s,t)$ to denote such a design.  

 
All computations involving factor levels are performed in the Galois field of order $s$, denoted by $GF(s)$. The set $GF(s)^n$ denotes all $n$-tuples with entries from $GF(s)$. When $s$ is a prime, the elements of $GF(s)$ are ${0,1,\ldots,s-1}$, with addition and multiplication defined modulo $s$.
When $s=p^n$ is a prime power ($n \geq 2$), the elements of $GF(s)$ can be represented as polynomials over $GF(p)$ modulo a primitive polynomial $h(x)$ of degree $n$. In this case, there exists a primitive element $\alpha$ satisfying $h(\alpha)=0$, and every nonzero element of $GF(s)$ can be expressed as a power of $\alpha$. Arithmetic operations in $GF(s)$ are then performed modulo $h(x)$.
For example, consider $GF(2^3)$ with primitive element $x$. One possible primitive polynomial is
$h(x)=x^3+x+1.$
Since $x^3=x+1$ modulo $h(x)$, every element of $GF(2^3)$ can be expressed in terms of powers of $x$, namely,
$
{0,1,x,x^2,x^3=x+1,x^4=x^2+x,x^5=x^2+x+1,x^6=x^2+1}.
$

A  $\mathbf{D} = OA(N,k,s,t)$ is called a {\em grouped orthogonal array} of  $N$ runs, $s$ levels, strength $t$ for $k$ factors if it can be divided into $m$ subarrays,  $\mathbf{D}_i=OA(N, k_i, s, t_i)$, such that $k=k_1+k_2+\cdots+k_m$ and $t_i > t$ for  $i=1,\ldots, m-1$ and $t_m \geq t$ \cite{chen2025}. Such a $\mathbf{D}$ is denoted as $GOA(N, (k_1, k_2, \ldots, k_m), (t_1, t_2, \ldots, t_m),s, t)$. 
If all subarrays $\mathbf{D}_i$'s have the same size $k_0$ and are of the same strength $t_0$, $\mathbf{D}$ is denoted as $GOA(N, k_0 \times m, t_0 \times m, s, t)$. The two notations
may be used in combination. 
Table \ref{table11} presents  
a $GOA(8,(4,3), 3\times 2,2,2)$. The columns of $\mathbf{D}$ can be partitioned into $\mathbf{D}_1=OA(8,4,2,3)$ and $\mathbf{D}_2=OA(8,3,2,3)$.

\renewcommand\arraystretch{0.45}
\setlength{\tabcolsep}{12pt}
\begin{table*}[htpb]
\setlength{\abovecaptionskip}{0.cm}
\setlength{\belowcaptionskip}{0.2cm}
\centering  
\caption{A $GOA(8,(4,3), 3\times 2,2, 2)$\label{table11}}
\begin{tabular}{cccc c ccc}
\hline
\multicolumn{4}{c}{$\mathbf{D}_1$}&&\multicolumn{3}{c}{$\mathbf{D}_2$}\\
\cline{1-4} \cline{6-8}
1&0&0&1&&0&1&1\\  
0&1&0&1&&1&0&1\\ 
0&0&1&1&&1&1&0\\ 
1&1&0&0&&1&1&0\\ 
1&0&1&0&&1&0&1\\ 
0&1&1&0&&0&1&1\\  
1&1&1&1&&0&0&0\\  
0&0&0&0&&0&0&0\\
\hline
\end{tabular}
\end{table*}

Another important concept we will use for constructing grouped orthogonal arrays is difference schemes \cite{seiden1954problem}.   An $r \times c$ matrix with entries from $GF(s)$ is called a {\em difference scheme} if the difference of any two columns contains all elements of $GF(s)$ equally often;  this matrix is commonly denoted by D$(r,c,s)$.  \cite{HEDAYAT1996307} generalized this concept and introduced {\em difference schemes of strength $t$}.  Let $\mathcal{M}^t$ be the set containing all distinct $t$-tuples defined on $GF(s)$. The set $\mathcal{M}^t$ can be partitioned into $s^{t-1}$ disjoint subsets, denoted as $\mathcal{M}_i^t$, for $i=0,\ldots, s^{t-1}-1$, where $\mathcal{M}_0^t$ consists of $s$ arbitrary elements from $\mathcal{M}^t$, and for $i=1,\ldots, s^{t-1}-1$, $\mathcal{M}_i^t$'s  are all distinct cosets of $\mathcal{M}_0^t$.
An $r\times c$ matrix $\mathbf{D}$ defined on $GF(s)$ is called a {\em difference scheme} of strength $t$ if in every $r \times t$ subarray of $\mathbf{D}$,  each coset $\mathcal{M}_0^t, \ldots,\mathcal{M}_{s^{t-1}-1}^t$ is represented equally often when the rows of the subarray are viewed as elements of  $\mathcal{M}^t$.  Such a difference scheme is denoted by   $D_t(r,c,s)$. 
For $t = 2$, this definition   reduces to the usual definition of difference schemes as introduced by \cite{seiden1954problem}. For example,  a  $D_3 (8, 7, 2)$ is 
\[
\scalebox{0.95}{$
\mathbf{D} =
\begin{pmatrix}
0&0&0&0&0&0&0\\  
0&0&1&0&1&1&1\\ 
0&1&0&1&0&1&1\\ 
0&1&1&1&1&0&0\\ 
1&0&0&1&1&0&1\\ 
1&0&1&1&0&1&0\\ 
1&1&0&0&1&1&0\\ 
1&1&1&0&0&0&1
\end{pmatrix}.
$}
\]
That is, among the rows of any $8 \times 3$ subarray of $\mathbf{D}$ are two occurrences of each of $\mathcal{M}_0^3$, $\mathcal{M}_1^3$, $\mathcal{M}_2^3$ and $\mathcal{M}_3^3$, where  $\mathcal{M}_0^3=\{(0,0,0), (1,1,1)\}, \mathcal{M}_1^3=\{(1,0,0), (0,1,1)\}, \mathcal{M}_2^3=\{(0,1,0), (1,0,1)\}$ and $\mathcal{M}_3^3=\{(0,0,1), (1,1,0)\}$.

We now review the operators {\em Kronecker product} and {\em generalized Kronecker product} that will be used in our constructions. Let $\mathbf{A}=[a_{ij}]_{n_1\times r_1}$ be an $n_1\times r_1$ matrix and $\mathbf{B}=[b_{ij}]_{n_2\times r_2}$ be an $n_2\times r_2$ matrix with the elements from $GF(s)$. 
Define an $(n_1n_2)\times (r_1r_2)$ matrix $\mathbf{A}\oplus \mathbf{B}$ to be the Kronecker product of $ \mathbf{A}$ and $\mathbf{B}$,
\begin{align}\label{kronecker}
    \mathbf{A}{ \oplus} \mathbf{B} &= \begin{bmatrix}
        a_{11}+ \mathbf{B} &\ldots&a_{1r_1}+\mathbf{B}\\
        \vdots&\ddots&\vdots\\
        a_{n_11}+\mathbf{B}&\ldots&a_{n_1r_1}+\mathbf{B}
    \end{bmatrix},
\end{align}  
\noindent where $+$ denotes the usual addition operation in Galois field. For an $n_1 \times m_1$ matrix $\bA$ and $n_2 \times m_2$ matrix $\bB$, $n_1 \leq n_2$, and the rows of $\bB$ partitioned into submatrices $\bB_1, ..., \bB_{n_1}$,   the {\em generalized Kronecker product}  \cite{2022He} is defined as 
\begin{equation}\label{eq:gks}
\bA \oast \bB = \bA \oast \left( \begin{array}{c}
     \bB_1  \\
     \vdots \\
     \bB_{n_1}
\end{array}\right) = [\ba_{i} \oplus\bB_{i}]_{1 \leq i \leq n_1}=
  \left(
  \begin{array}{c}
    \ba_{1} \oplus \bB_{1} \\
    \vdots \\
    \ba_{n_1} \oplus\bB_{n_1}
  \end{array}
\right),
\end{equation}
where $\ba_i$ denotes the $i$th row of $\bA$ and the operator $\oplus$ represents the Kronecker product  
given in (\ref{kronecker}).  
Note that $\bA \oast \bB$ depends on the partition of the rows of $\bB$, and it  has the same number of runs as $\bB$,   the number $m_1m_2$ of factors. For example, with addition modulo 3, let $\bA = (0,1,2)^T$ and 
\begin{center}
$\bB = \left(
\begin{array}{r}
\bB_1\\
\bB_2\\
\bB_3
\end{array}
\right)  =\left(\begin{array}{rrrr}
0 & 0 & 0 & 0 \\
1 & 2 & 2 & 0 \\
2 & 1 & 1 & 0 \\
\hdashline
0 & 2 & 1 & 1 \\
1 & 1 & 0 & 1 \\
2 & 0 & 2 & 1 \\
\hdashline
0 & 1 & 2 & 2 \\
1 & 0 & 1 & 2 \\
2 & 2 & 0 & 2
\end{array} 
\right)$, then we have 
$\bA \oast \bB =  \left( \begin{array}{rrrr}
0 & 0 & 0 & 0 \\
1 & 2 & 2 & 0 \\
2 & 1 & 1 & 0 \\

1 & 0 & 2 & 2 \\
2 & 2 & 1 & 2 \\
0 & 1 & 0 & 2 \\

2 & 0 & 1 & 1 \\
0 & 2 & 0 & 1 \\
1 & 1 & 2 & 1
\end{array} 
\right)$.
\end{center}

\section{Construction Using Orthogonal Arrays}\label{sec3_1}

In this section, we present three construction methods for grouped orthogonal arrays with groups of strength three based on orthogonal arrays of smaller sizes. Theoretical results for the proposed constructions are established, and  new grouped orthogonal arrays that cannot be obtained using existing methods are provided.


Construction~\ref{construction1} below uses two arrays
$\bA$ and $\bB$ where $\bA=OA(n_1, m_1, s, 2)$ for $m_1 \ge 2$
and $\bB=OA(n_2, m_2, s, 2)$ for $m_2=2$ or an $OA(n_2, m_2, s, 3)$ for $m_2\geq 3$, where $s \ge 3$ is a prime power. Here, the entries of $\bA$ and $\bB$ are from the Galois field $GF(s)=\{\alpha_0,\ldots,\alpha_{s-1}\}$ with $\alpha_0=0.$ Let $m_1^* = \lceil m_1/2 \rceil$.

\begin{const}\label{construction1}
\  \\ 
The columns of $\bA$ can be partitioned as $\bA = ( \bA^{(1)}, \ldots, \bA^{(m_1^*)})$, where each of $\bA^{(1)}, \ldots, \bA^{{ (\lfloor m_1/2\rfloor)}}$ contains two columns. For $i=1,\ldots, n_1$, let $\ba_i^{(j)}$ denote the $i$th row of each $\bA^{(j)}$.
For $h=1,\ldots,s-1$, and $j=1,\ldots, m_1^*$, define
\begin{equation*}\label{eq:construction3}
\bD_{h}^{(j)}  =  \bA^{(j)} \oast ( \alpha_h * \bB) = \left[ \begin{array}{c}
                                                \ba_{1}^{(j)}  \oplus  ( \alpha_h * \bB)\\
                                                \vdots \\
                                                \ba_{n_1}^{(j)} \oplus  ( \alpha_h * \bB)
                                              \end{array}\right ],
\end{equation*}
and  
\begin{equation*}\label{eq:construction3s}
\bD_{s}  =  \bzero_{n_1} \oast   \bB = \left[ \begin{array}{c}
                                                0 \oplus     \bB\\
                                                \vdots \\
                                                0 \oplus   \bB
\end{array}\right ],
\end{equation*}
\noindent where  $*$  represents the multiplication in a field and $\bzero_{n_1}$ denotes a column vector with $n_1$ zeros.    Now obtain an array
\begin{equation*}\label{eq:E2}
\bE=
    \left [\bD_{1}^{(1)},\bD_2^{(1)}, \ldots,  \bD_{s-1}^{(1)}, \ldots, \bD_{1}^{(m_1^*)},\bD_2^{(m_1^*)}, \ldots,  \bD_{s-1}^{(m_1^*)},\bD_{s} \right ].
\end{equation*} 
 
\noindent If $s$ is odd, let $\bE_{(j-1)((s-1)/2)+f} = \left[ \bD_{2f-1}^{(j)},\bD_{2f}^{(j)} \right]$ for $j=1,\ldots, m_1^*$  and $f = 1, \ldots, (s-1)/2$,  and $\bE_{m_1^*(s-1)/2+ 1} = \bD_{s}$.  If $s$ is even, for $j=1, \ldots, m_1^*$, let $\bE_{ (j-1)(s/2)+f} = \left [ \bD_{2f-1}^{(j)},\bD_{2f}^{(j)} \right]$  and $f =1, \ldots,  s/2-1$,  $\bE_{ (j-1)(s/2)+ s/2} = \bD_{s-1}^{(j)}$, and $\bE_{ m_1^*s/2+ 1} = \bD_{s} $.
\end{const}

Proposition~\ref{proposition1} summarizes the property of $\bE$ in  Construction~\ref{construction1} when $\bB$ is an orthogonal array of strength three.  Results corresponding to $\bB $ of strength two and $m_2 =2 $ can be derived similarly.  Example~\ref{example1} gives an illustration of Construction~\ref{construction1}.

\begin{proposition}
\label{proposition1}
Let $\bA$ be an $OA(n_1, m_1, s, 2)$ for $m_1 \ge 2$ and $\bB$ be 
an $OA(n_2, m_2, s, 3)$ for $m_2\geq 3$, where $s \ge 3$ is a prime power, then we have the following results.
\begin{itemize}
\item[(i)]  If \( s \) is odd and $m_1$ is even, then $\bE$ in  Construction \ref{construction1}   is a  
{\footnotesize 
          \begin{align*}
          GOA\bigg(n_1n_2, \Big(4m_2\times \frac{m_1(s-1)}{4}, m_2\Big), 3 \times \frac{m_1(s-1)+4}{4}, s, 2\bigg);
          \end{align*}
          }
\item[(ii)] If \( s \) is even and $m_1$ is even, then \( \boldsymbol{E} \)  in  Construction \ref{construction1}   is a
{\footnotesize 
        \begin{align*}
          GOA\bigg(n_1n_2, \Big(4m_2 \times \frac{m_1(s-2)}{4}, 2m_2 \times \frac{m_1}{2},m_2 \Big), 3 \times \frac{m_1s+4}{4}, s, 2\bigg);
        \end{align*} 
        }
\item[(iii)]  If \( s \) is odd and $m_1$ is odd, then $\bE$ in  Construction \ref{construction1}   is a  
{\footnotesize 
          \begin{align*}
          GOA\bigg(n_1n_2, \Big(4m_2\times \frac{(m_1-1)(s-1)}{4}, 
          2m_2 \times \frac{s-1}{2},m_2\Big), 3 \times  \frac{(m_1+1)(s-1)+4}{4}, s, 2\bigg);
          \end{align*}
}
\item[(iv)] If \( s \) is even and $m_1$ is odd, then \( \boldsymbol{E} \) in  Construction \ref{construction1}   is a  
 {\footnotesize 
         $$ GOA\bigg(n_1n_2, \Big(4m_2 \times \frac{(m_1-1)(s-2)}{4}, 
          2m_2 \times \frac{m_1+s-3}{2},m_2\times 2\Big), 3 \times  \frac{(m_1+1)s+4}{4}, s, 2\bigg).$$
}
\end{itemize}
\end{proposition}

\begin{example}\label{example1}
Consider the case $s=3$. A $GOA(243,(16\times2,4),3\times3,3,2)$ is obtained from an $\boldsymbol{A} = OA(9,4,3,2)$  and a $\boldsymbol{B} = OA(27,4,3,3)$ following part (i) of Proposition~\ref{proposition1};  if  $\boldsymbol{A} = OA(27,13,3,2)$ instead, we get  a $GOA(729,(16\times6,8,4),3\times 8,3,2)$ by part (iii) of the proposition. 

     \end{example}

Table~\ref{table22} presents several grouped orthogonal arrays with large group sizes constructed using Construction~\ref{construction1}. For run sizes  729, 1024, and 3125, Construction~\ref{construction1} yields grouped orthogonal arrays with group sizes  40, 24, and 24, respectively, whereas the largest group sizes reported in the existing literature \cite{chen2025,li2025grouped} are only 32, 16, and 16, respectively.
This demonstrates the effectiveness of the proposed approach in producing grouped orthogonal arrays with larger group sizes.

 \renewcommand\arraystretch{0.66}
\setlength{\tabcolsep}{7.7pt}
\begin{table}[htpb]
\setlength{\abovecaptionskip}{0.cm}
\setlength{\belowcaptionskip}{0.2cm}
\centering 
\caption{Grouped orthogonal arrays from Construction~\ref{construction1}.\label{table22}}
\begin{tabular}{lll}
\hline
$\bA$ & $\bB$& Constructed grouped orthogonal array\\
\hline
{$OA(9,4,3,2)$} & {$OA(54,5,3,3)$} &$GOA(486, (20 \times 2, 5), 3 \times 3, 3, 2)$\\

{$OA(9,4,3,2)$} & {$OA(81,10,3,3)$} &$GOA(729, (40 \times 2, 10), 3 \times 3, 3, 2)$\\

{$OA(16,5,4,2)$} & {$OA(64,6,4,3)$} &$GOA(1024, (24 \times 2, 12\times 3, 6\times 2), 3 \times 7, 4, 2)$\\

{$OA(25,6,5,2)$} & {$OA(125,6,5,3)$} &$GOA(3125, (24 \times 6, 6), 3 \times 7, 5, 2)$\\

\hline
\end{tabular}
\end{table}

Note that in Construction~\ref{construction1}, $\boldsymbol{B}$ is either a two-column orthogonal array of strength two or an orthogonal array of strength three with more than two columns. We now introduce Construction~\ref{construction2} that works for an orthogonal array $\bB$ of strength two with more than two columns by  repeatedly applying  Construction~\ref{construction1} to two-column blocks of $\boldsymbol{B}$.

\begin{const} \label{construction2}
Let $\bA$ be an $n_1 \times m_1$ matrix, and let $\bB$ be an $OA(n_2, m_2, s, 2),$ where $m_2=2f.$
Partition the 
columns of 
$\bB$ as $(\bB_1,\ldots,\bB_f)$, where each $\bB_i$ consists of two columns. For $i=1,\ldots,f$, apply Construction~\ref{construction1} to $\bA$ and $\bB_i$ to obtain an array $\bE_i$. Let $\bE$ denote the array obtained by concatenating all columns of $\bE_1,\ldots,\bE_f$.
\end{const}

Proposition~\ref{proposition2} summarizes the property of $\bE$ in  Construction~\ref{construction2}  and Example~\ref{example2} illustrates the construction procedure.

\begin{proposition}
\label{proposition2}
Let $\bA$ be an $OA(n_1,m_1,s,2)$ with $m_1 \ge 2$ and $\bB$ be an $OA(n_2,m_2,s,2)$ where $m_2=2f,$
then we have the following result.
\begin{itemize}
\item[(i)]  If \( s \) is odd and $m_1$ is even, then $\bE$ in Construction \ref{construction2}    is a 
{\footnotesize 
          \begin{align*}
          GOA\bigg(n_1n_2, \Big(8 \times \frac{fm_1(s-1)}{4}, m_2 \Big), \Big (3 \times  \frac{fm_1(s-1)}{4},2\Big), s, 2\bigg);
          \end{align*} 
          }
\item[(ii)] If \( s \) is even and $m_1$ is even, then \( \boldsymbol{E} \)  in Construction \ref{construction2}  is a  
  {\footnotesize       \begin{align*}
          GOA\bigg(n_1n_2, \Big( 8 \times \frac{fm_1(s-2)}{4}, 4 \times \frac{fm_1}{2},m_2\Big), \Big(3 \times  \frac{fm_1s}{4},2\Big), s, 2\bigg);
        \end{align*}  
        }
\item[(iii)]  If \( s \) is odd and $m_1$ is odd, then $\bE$ in Construction \ref{construction2}    is a 
{\footnotesize 
          \begin{align*}
          GOA\bigg(n_1n_2, \Big(8 \times \frac{f(m_1-1)(s-1)}{4}, 
          4 \times \frac{f(s-1)}{2},m_2\Big), \Big(3 \times  \frac{f(m_1+1)(s-1)}{4},2\Big ), s, 2\bigg);
          \end{align*} 
}
\item[(iv)] If \( s \) is even and $m_1$ is odd, then \( \boldsymbol{E} \)  in Construction \ref{construction2}   is a  
{\footnotesize  
         $$ \hspace{-0.5cm} GOA\bigg(n_1n_2, \Big(8 \times \frac{f(m_1-1)(s-2)}{4}, 
          4 \times \frac{f(m_1+s-3)}{2},2m_2\Big), \Big(3 \times  \frac{f((m_1+1)s-4)}{4},2\Big), s, 2\bigg).$$
}
 
\end{itemize}
\end{proposition}

\begin{example}\label{example2}
Using $\boldsymbol{A} = OA(27,13,3,2)$ and $\boldsymbol{B} = OA(9,4,3,2)$, Construction~\ref{construction2} gives  a $GOA(243,(8\times12,4\times2,4),(3\times14,2),3,2)$; and a $GOA(625,(8\times18,6),(3\times18,2),5,2)$ is obtained from an $\boldsymbol{A}= OA(25,6,5,2)$  and a $\boldsymbol{B}= OA(25,6,5,2)$. 
\end{example}

\begin{remark}\label{remark1}
We present Construction~\ref{construction2} and Proposition~\ref{proposition2} under the assumption that $m_2$ is even for simplicity of exposition. However, the idea of the construction also applies when $m_2$ is odd. In this case, let $\bB_{f+1}=\boldsymbol{b}_{m_2}$ denote the last column of $\bB$, and apply Construction~\ref{construction1} to $\bA$ and $\bB_{f+1}$ to obtain a grouped orthogonal array $\bE_{f+1}$. Compared with each $\bE_i$, $i=1,\ldots,f$, the array $\bE_{f+1}$ has the same number of groups, but each group contains only half as many columns.
For example, let $\bA$ be an $OA(25,6,5,2)$ and $\bB$ be an $OA(125,31,5,2)$. Partition the first 30 columns of $\bB$ into $(\bB_1,\ldots,\bB_{15})$, where each $\bB_i$ consists of two columns. For $i=1,\ldots,15$, applying Construction~\ref{construction1} to $\bA$ and $\bB_i$ yields an array $\bE_i = GOA(3125,(8\times 6,2),(3\times 6,2),5,2)$. By Proposition~\ref{construction2}, concatenating $(\bE_1,\ldots,\bE_{15})$ gives a $GOA(3125,(8\times 90,30),(3\times 90,2),5,2)$.
Next, let $\bB_{16}=\boldsymbol{b}_{31}$ and apply Construction~\ref{construction1} to $\bA$ and $\bB_{16}$, producing a grouped orthogonal array $\bE_{16}= GOA(3125,(4\times 6,1),(3\times 6,1),5,2)$. Concatenating the columns of $\bE_1,\ldots,\bE_{16}$ then yields a $GOA(3125,(8\times 90,4\times 6,31),(3\times 96,2),5,2)$.
\end{remark}

\renewcommand\arraystretch{0.66}
\setlength{\tabcolsep}{4.5pt}
\begin{table}[htpb]
\setlength{\abovecaptionskip}{0.cm}
\setlength{\belowcaptionskip}{0.2cm}
\centering 
\begin{threeparttable}
\caption{Grouped orthogonal arrays from Construction~\ref{construction2}.\label{table33}}
\begin{tabular}{lll}
\hline
$\bA$ & $\bB$& {Constructed grouped orthogonal Array}\\
\hline

    {$OA(18,7,3,2)$} & {$OA(9,4,3,2)$} &
 $GOA(162, (8 \times 6, 4\times 2,4), (3 \times 8,2), 3, 2)^\dag$\\

    {$OA(27,13,3,2)$} & {$OA(9,4,3,2)$} &
 $GOA(243, (8 \times 12, 4\times 2,4), (3 \times 14,2), 3, 2)$\\
 
     {$OA(54,25,3,2)$} & {$OA(9,4,3,2)$} &
 $GOA(486, (8 \times 24, 4\times 2,4), (3 \times 26,2), 3, 2)^\dag$\\
 
  {$OA(9,4,3,2)$} & {$OA(81,40,3,2)$} &
 $GOA(729, (8 \times 40,40), (3 \times 40,2), 3, 2)$\\

  {$OA(25,6,5,2)$} & {$OA(25,6,5,2)$} &
 $GOA(625, (8 \times 18,6), (3 \times 18,2), 5, 2)$\\

   {$OA(25,6,5,2)$} & {$OA(125,30,5,2)$} &
 $GOA(3125, (8 \times 90,30), (3 \times 90,2), 5, 2)^\dag$\\
    {$OA(25,6,5,2)$} & {$OA(125,31,5,2)$} &
 $GOA(3125,(8\times 90,4\times 6,31),(3\times 96,2),5,2)^{\dag *}$\\
 \hline
\end{tabular}
\begin{tablenotes}
\scriptsize{
\item {$*$ The corresponding grouped orthogonal array is obtained as described in Remark \ref{remark1}.\\
 $\dag$ The corresponding grouped orthogonal array is new and cannot be
generated by existing methods}}

\end{tablenotes}
\end{threeparttable}
\end{table}

Table~\ref{table33} summarizes some grouped orthogonal arrays generated by Construction~\ref{construction2}. Some of these arrays, marked by $\dag$, are new and cannot be obtained using existing methods. Although existing methods can generate grouped orthogonal arrays with the same design parameters (run size, number of groups, and group sizes) as some of the remaining arrays, the resulting designs are not necessarily isomorphic to ours.


We now present Construction~\ref{construction3}, which repeatedly applies Construction~\ref{construction1} using  $\boldsymbol{B}$ that is composed of strength-three orthogonal array blocks, i.e., $\boldsymbol{B}$ is a grouped orthogonal array consisting of strength-three groups.

\begin{const} \label{construction3}
Let $\bA$ be an $n_1 \times m_1$ matrix, and $\bB$ be a $GOA(n_2, m_2 \times f, 3 \times f, s, 2)$.
Partition $\bB$ as $\bB = (\bB_1,\ldots,\bB_f)$
where each $\bB_i$ is an $OA(n_2, m_2, s, 3)$. For $i=1,\ldots,f$, apply Construction~\ref{construction1} to $(\bA,\bB_i)$ to obtain an array $\bE_i$. Let $\bE$ denote the array obtained by concatenating all columns of $\bE_1,\ldots,\bE_f$.
\end{const}

Proposition~\ref{proposition3} summarizes the property of $\bE$ in  Construction~\ref{construction3}  and Example~\ref{example3} illustrates the procedure of the construction. 

\begin{proposition}
\label{proposition3}
Let $\bA$ be an $n_1 \times m_1$ matrix and $\bB$ be a $GOA(n_2, m_2 \times f, 3 \times f, s, 2)$, then we have the following results.
\begin{itemize}
\item[(i)]  If \( s \) is odd and $m_1$ is even, then $\bE$ in  Construction \ref{construction3}   is a 
{\footnotesize
          \begin{align*}
          GOA\bigg(n_1n_2, \Big( 4m_2 \times \frac{fm_1(s-1)}{4}, m_2\times f \Big), 3 \times  \frac{f(m_1(s-1)+4)}{4}, s, 2\bigg);
          \end{align*}\vspace*{-0.6in}  }
\item[(ii)] If \( s \) is even and $m_1$ is even, then \( \boldsymbol{E} \)  in  Construction \ref{construction3}  is a  
{\footnotesize 
        \begin{align*}
          GOA\bigg(n_1n_2, \Big( 4m_2 \times \frac{fm_1(s-2)}{4}, 2m_2 \times \frac{fm_1}{2},m_2\times f \Big), 3 \times  \frac{f(m_1s+4)}{4}, s, 2\bigg);
        \end{align*}  
        }
\item[(iii)]  If \( s \) is odd and $m_1$ is odd, then $\bE$ in  Construction \ref{construction3}   is a  
{\footnotesize 
         $$\hspace{-1cm}   GOA\bigg(n_1n_2, \Big( 4m_2 \times \frac{f(m_1-1)(s-1)}{4}, 
          2m_2 \times \frac{f(s-1)}{2},m_2\times f\Big), 3 \times \frac{f((m_1+1)(s-1)+4)}{4}, s, 2\bigg);$$
}
\item[(iv)] If \( s \) is even and $m_1$ is odd, then \( \boldsymbol{E} \) in  Construction \ref{construction3}   is a  
{\footnotesize 
         $$\hspace{-1.1cm} GOA \bigg(n_1n_2, \Big(4m_2 \times \frac{f(m_1-1)(s-2)}{4}, 
          2m_2 \times \frac{f(m_1+s-3)}{2},m_2\times 2f \Big), 3 \times  \frac{f((m_1+1)s+4)}{4}, s, 2\bigg).$$
}
\end{itemize}
\end{proposition}

\begin{example}\label{example3}
A $GOA(243,(16\times6,4\times3),3\times9,3,2)$ is obtained from an $\boldsymbol{A} = OA(9,4,3,2)$ and a $\boldsymbol{B}= GOA(27,4\times 3,3\times 3,3,2)$. 

\end{example}

\renewcommand\arraystretch{0.6}
\setlength{\tabcolsep}{2.6pt}
\begin{table}[htpb]
\setlength{\abovecaptionskip}{0.cm}
\setlength{\belowcaptionskip}{0.2cm}
\centering 
\small
\caption{Grouped orthogonal arrays from Construction~\ref{construction3}.\label{table44}}
\begin{tabular}{lll}
\hline
$\bA$ & $\bB$& {Constructed grouped orthogonal array}\\
\hline
 {$OA(9,4,3,2)$} & {$GOA(27,4\times 3,3\times 3,3,2)$} &
 $GOA(243, (16 \times 6,4\times 3), 3 \times 9, 3, 2)$\\

  {$OA(18,7,3,2)$} & {$GOA(27,4\times 3,3\times 3,3,2)$} &
 $GOA(486, (16 \times 9, 8\times 3, 4\times 3), 3 \times 15, 3, 2)$\\

   {$OA(27,13,3,2)$} & {$GOA(27,4\times 3,3\times 3,3,2)$} &
 $GOA(729, (16 \times 18, 8\times 3, 4\times 3), 3 \times 24, 3, 2)$\\

 {$OA(9,4,3,2)$} & {$GOA(81,8\times 4,3\times 4,3,2)$} &
 $GOA(729, (32 \times 8, 8\times 4), 3 \times 12, 3, 2)$\\

 {$OA(25,6,5,2)$} & {$GOA(125,4\times 6,3\times 6,5,2)$} &
 $GOA(3125, (16 \times 36, 4\times 6), 3 \times 42, 5, 2)$\\
 \hline
\end{tabular}
\end{table}

Table~\ref{table44} summarizes some grouped orthogonal arrays obtained using Construction~\ref{construction3}. Notably, all of these grouped orthogonal arrays are new and cannot be generated by existing methods. 


\section{Construction Using Difference Schemes}\label{sec3_2}

Constructions~\ref{construction1}--\ref{construction3}, together with the methods of \cite{chen2025} and \cite{li2025grouped}, provide several approaches for constructing grouped orthogonal arrays with group strength three. However, these methods neither extend to group strengths of four or higher nor produce grouped orthogonal arrays with overall strength three. In general, they do not yield designs with group strength greater than three and overall strength two. The only known exception is a construction of \cite{chen2025}, which produces regular grouped orthogonal arrays with group strength at least four and overall strength two.
In this section, we introduce new constructions based on difference schemes of strength three or higher. The proposed Constructions~\ref{construction4}--\ref{construction6} further extend the framework by enabling the construction of grouped orthogonal arrays with higher within-group strength and overall strength two or three.

Construction~\ref{construction4} below combines an orthogonal array of strength $t_1 \geq 2$ with a difference scheme of strength $t_2 \geq 3$ to construct grouped orthogonal arrays with within-group strength $t_2$. The property of designs in the construction is summarized in Proposition~\ref{proposition4} and an example of the construction is given in Example~\ref{example4}.


\begin{const}\label{construction4}
    Let \( \boldsymbol{A} = (\boldsymbol{a}_1, \dots, \boldsymbol{a}_{m_1})\) be an  \( OA(n_1, m_1, s, t_1) \) with \( t_1  = 2 \) or \( t_1  = 3 \), and let \( \boldsymbol{B} \) be a \( D_{t_2}(n_2, m_2, s) \) with \(t_2 \geq 3 \), where \( \boldsymbol{a}_i \)  denotes the $i$th column of $\boldsymbol{A}$.  For $i = 1, \dots, m_1$,  let $ \boldsymbol{E}_{i} = \boldsymbol{a}_{i} \oplus \boldsymbol{B}$
\noindent and define
    \begin{align}\label{eq:construction4}
        \boldsymbol{E} =   
(            \boldsymbol{E}_1,   \ldots,  \ \boldsymbol{E}_{m_1})  . 
    \end{align}     
\end{const}

\begin{proposition}
\label{proposition4}
 Let \( \boldsymbol{A} \) be an  \( OA(n_1, m_1, s, t_1) \) with \( t_1  = 2 \) or \( t_1  = 3 \) and let \( \boldsymbol{B} \) be a \( D_{t_2}(n_2, m_2, s) \) with  \( t_2 \geq 3 \). Then \( \boldsymbol{E} \) in (\ref{eq:construction4}) is a \( GOA(n_1n_2, m_2 \times m_1, t_2 \times m_1, s, t_1) \).  
\end{proposition}

\begin{example}\label{example4}
Let \( \boldsymbol{A} \)  be an \( OA(27, 13, 3, 2) \) and let \( \boldsymbol{B} \) be a   \( D_4(27, 5, 3) \).
Then a \( GOA(729, 5 \times 13, 4 \times 13, 3, 2) \) can be obtained by Construction~\ref{construction4}.  
\end{example}

Construction~\ref{construction4} employs orthogonal arrays of strength two or three together with difference schemes of strength $t_2 \ge 3$. The former can be obtained from websites such as \cite{sloane} and \cite{EendebakSchoen} (see also the \texttt{oapackage} by \cite{Eendebak2019}). The latter can be constructed using various methods introduced in \cite{HEDAYAT1996307}. We implement these methods and provide a catalog of difference schemes of strength three and higher, which is in the website \url{https://github.com/devonlin/Grouped-Orthogonal-Arrays-and-Difference-Schemes-of-Higher-Strength}.  Table \ref{table55} summarizes some grouped orthogonal arrays generated by Construction~\ref{construction4}.  Most of these designs are new because they achieve within-group strength 4. While \cite{chen2025} can  produce grouped orthogonal arrays with larger group sizes or more groups in a few cases, those constructions are limited to regular designs. In contrast, our method can generate both regular and nonregular designs.

\renewcommand\arraystretch{0.6}
\setlength{\tabcolsep}{13pt}
\begin{table}[htpb]
\setlength{\abovecaptionskip}{0.cm}
\setlength{\belowcaptionskip}{0.2cm}
\centering  
\caption{\small {Grouped orthogonal arrays from Construction \ref{construction4}}.\label{table55}}
\begin{tabular}{lll}
\hline
\textbf{A} & \textbf{B} & {Constructed grouped orthogonal array}  \\
\hline
$OA(9, 4, 3, 2)$ & $D_4(27, 5, 3)$& $GOA(243, 5\times 4, 4\times 4, 3, 2)$\\
  & $D_4(81, 9, 3)$& $GOA(729, 9\times 4, 4\times 4, 3, 2)$\\
$OA(27, 4, 3, 3)$ & $D_4(27, 5, 3)$& $GOA(729, 5\times 4, 4\times 4, 3, 3)$\\
  & $D_4(81, 9, 3)$& $GOA(2187, 9\times 4, 4\times 4, 3, 3)$ \\
$OA(27, 13, 3, 2)$ & $D_4(27, 5, 3)$& $GOA(729, 5\times 13, 4\times 13, 3, 2)$\\
&$D_4(81, 9, 3)$& $GOA(2187, 9\times 13, 4\times 13, 3, 2)$\\

$OA(16, 5, 4, 2)$ & $D_4(64, 5, 4)$& $GOA(1024, 5\times 5, 4\times 5, 4, 2)$\\
 & $D_4(256, 10, 4)$& $GOA(4096, 10\times 5, 4\times 5, 4, 2)$\\
  $OA(32, 9, 4, 2)$ & $D_4(64, 5, 4)$& $GOA(2048, 5\times 9, 4\times 9, 4, 2)$\\
 & $D_4(256, 10, 4)$& $GOA(8192, 10\times 9, 4\times 9, 4, 2)$\\
$OA(64, 6, 4, 3)$ & $D_4(64, 5, 4)$& $GOA(4096, 5\times 6, 4\times 6, 4, 3)$\\
$OA(25, 6, 5, 2)$& $D_4(125, 5, 5)$& $GOA(3125, 5\times 6, 4\times 6, 5, 2)$\\
 $OA(50, 11, 5, 2)$&  $D_4(125, 5, 5)$& $GOA(6250, 5\times 11, 4\times 11, 5, 2)$\\
 \hline
\end{tabular}
\end{table}

In Construction~\ref{construction4}, grouped orthogonal arrays are constructed from an orthogonal array and a difference scheme, where the within-group strength is determined by the difference scheme and the overall strength is determined by the orthogonal array. We next develop a construction in which the within-group strength is determined by the orthogonal array, while the overall strength depends on the difference scheme.

\begin{const}\label{construction5}
Let \( \boldsymbol{A} = (\boldsymbol{a}_1, \dots, \boldsymbol{a}_{m_1})\) and \( \boldsymbol{B} = (\boldsymbol{b}_1, \dots, \boldsymbol{b}_{m_2})\) with $m_2 > 2$, where \( \boldsymbol{a}_i \)  and  \( \boldsymbol{b}_i  \) denote the \( i \)th column of \( \boldsymbol{A} \) and \( \boldsymbol{B} \), respectively. Let $g= \hbox{min}\{m_1,m_2-1\}$. For $ j = 1, \dots, g$,  define   
$\boldsymbol{E}_{j} = (\boldsymbol{A} \oplus \boldsymbol{b}_{j}, \boldsymbol{a}_{j} \oplus \boldsymbol{b}_{m_2}),$  and obtain
    \begin{align}
        \boldsymbol{E} =  
            (\boldsymbol{E}_1,  \ldots, \boldsymbol{E}_{g})
    .\label{eq:construction5}
    \end{align}  
\end{const}

\begin{proposition}\label{proposition5}
Let \( \boldsymbol{A}\) be an \( OA(n_1, m_1, s, t_1) \) with $t_1\geq 3$ and  \( \boldsymbol{B} \) be a \( D_{t_2}(n_2, m_2, s) \) with 
$t_2=2$ or $t_2=3.$
Then $\boldsymbol{E}$ in \eqref{eq:construction5} is a \( GOA(n_1n_2, (m_1+1) \times g, t_1 \times g, s, t_2)\).
Moreover, if $m_1=t_1$, then $\boldsymbol{E}$ in  \eqref{eq:construction5} is a \( GOA(n_1n_2, (m_1+1) \times g, (t_1+1) \times g, s, t_2)\). 
\end{proposition}

Proposition~\ref{proposition5} shows that, in general, $\boldsymbol{E}$ is a grouped orthogonal array consisting of $g$ groups, each of strength $t_1$. Moreover, when $m_1 = t_1$, i.e., $\boldsymbol{A}$ is a full factorial, each $\boldsymbol{E}_g$ becomes a full factorial  and thus of strength $t_1+1$. Example~\ref{example5} provides an illustration of Construction~\ref{construction5}.

\begin{example}\label{example5}
Let $\boldsymbol{A} = OA(81, 5, 3, 4)$ and $\boldsymbol{B} = D_3(9, 4, 3)$. Then $t_1 = 4$, $t_2 = 3$, and $g = \min\{m_1, m_2 - 1\} = 3$. By Proposition~\ref{proposition5}, we obtain a GOA(729,6×3,4×3,3,3).
Let $\boldsymbol{A} = OA(27, 3, 3, 3)$ and $\boldsymbol{B} = D(3, 3, 3)$. Then $t_1 = m_1 = 3$, $m_2 > t_2 = 2$, and $g = 2$. By Proposition~\ref{proposition5}, we obtain a GOA(81,4×2,4×2,3,2).
\end{example}

Table~\ref{table66} summarizes some grouped orthogonal arrays constructed via Construction~\ref{construction5}.   Comparing Table~\ref{table55} with Table~\ref{table66}, we observe that they yield different numbers of groups and group sizes, while Construction~\ref{construction5} offers greater flexibility in the run sizes of the resulting grouped orthogonal arrays because difference schemes exist for the run sizes that are not multiple of $s^3$.  For example,  Construction~\ref{construction5} can provide grouped orthogonal arrays of strength-four groups with 162 runs while 
Construction~\ref{construction4} cannot. 

 \renewcommand\arraystretch{0.6}
\setlength{\tabcolsep}{14pt}
\begin{table}[htpb]
\setlength{\abovecaptionskip}{0.cm}
\setlength{\belowcaptionskip}{0.2cm}
\centering  
\caption{\small {Grouped orthogonal arrays from Construction \ref{construction5}}.\label{table66}}
\begin{tabular}{lll}
\hline
\textbf{A} & \textbf{B} & {Constructed grouped orthogonal array}  \\
\hline
$OA(27, 3, 3, 3)$ & $D(3, 3, 3)$& $GOA(81, 4\times 2, 4\times 2, 3, 2)$\\
& $D(6, 6, 3)$& $GOA(162, 4\times 3, 4\times 3, 3, 2)$\\
& $D_3(9, 4, 3)$& $GOA(243, 4\times 3, 4\times 3, 3, 3)$\\
 
 $OA(81, 5, 3, 4)$ & $D_3(9, 4, 3)$& $GOA(729, 6\times 3, 4\times 3, 3, 3)$\\
& $D_3(18, 5, 3)$& $GOA(1458, 6\times 4, 4\times 4, 3, 3)$\\
& $D_3(27, 9, 3)$& $GOA(2187, 6\times 5, 4\times 5, 3, 3)$\\

{$OA(256, 5, 4, 4)$} & $D(4, 4, 4)$& $GOA(1024, 6\times 3, 4\times 3, 4, 2)$\\
& $D(8, 8, 4)$& $GOA(2048, 6\times 5, 4\times 5, 4, 2)$\\
& $D_3(16, 6, 4)$& $GOA(4096, 6\times 5, 4\times 5, 4, 3)$\\
 
$OA(625, 6, 5, 4)$ & $D(5, 5, 5)$& $GOA(3125, 7\times 4, 4\times 4, 5, 2)$\\
 & $D(10, 10, 5)$& $GOA(6250, 7\times 6, 4\times 6, 5, 2)$\\
 & $D_3(25, 6, 5)$& $GOA(15625, 7\times 5, 4\times 5, 5, 3)$\\
 \hline
\end{tabular}
\end{table}

Next, we present Construction~\ref{construction6}, a recursive method for constructing grouped orthogonal arrays with groups of strength three and overall strength two. This construction is based on smaller grouped orthogonal arrays with groups of strength three together with difference schemes of strength three. Unlike the recursive method in Construction~\ref{construction3}, which is based on orthogonal arrays, Construction~\ref{construction6} employs difference schemes. The idea of constructing strength-three orthogonal arrays from smaller arrays using difference schemes of strength three was first introduced by \cite{mukhopadhyay1981construction}; here, we extend this approach to grouped orthogonal arrays. Construction~\ref{construction6} differs from the recursive constructions in \cite{chen2025} and \cite{li2025grouped} in that \cite{chen2025} employs difference schemes of strength two in a different manner, whereas \cite{li2025grouped} does not use difference schemes at all.

\begin{const}\label{construction6}
Let $ \boldsymbol{A}= (\boldsymbol{A}_1, \ldots, \boldsymbol{A}_g)$ be a $ GOA(N, (m_1, \dots, m_g), 3\times g, s, 2)$ and  $\boldsymbol{B}$ be a $D_3(r, c, s),$  where  $\boldsymbol{A}_j$ is the $j$th group of $\boldsymbol{A}$. For \( j = 1, 2, \ldots, g \),  define \( \boldsymbol{E}_j = \boldsymbol{A}_j \oplus \boldsymbol{B} \) and obtain    
    \begin{align}\label{eq3}
        \boldsymbol{E} =(\boldsymbol{E}_1, \dots, \boldsymbol{E}_g). 
        \end{align}
\end{const}

Proposition~\ref{proposition6} summarizes the property of $\bE$ in  Construction~\ref{construction6}  and Example~\ref{example6} illustrates the  construction.

\begin{proposition}\label{proposition6}
Let $ \boldsymbol{A}$ be a $ GOA(N, (m_1, \dots, m_g), 3\times g, s, 2)$  and  $\boldsymbol{B}$ be a $D_3(r, c, s)$.
Then  \( \boldsymbol{E} \) in (\ref{eq3}) is a \( GOA(Nr, (cm_1, \dots, cm_g), 3 \times g, s, 2) \).
\end{proposition}

\begin{example}\label{example6}
Let $\bA$ be a $GOA(27,4\times3,3\times3,3,2)$ and let $\bB$ be a $D_3(27,9,3)$. Then Construction~\ref{construction6} yields a $GOA(729,36\times3,3\times3,3,2)$. As  pointed in the discussion for Construction~\ref{construction1}, for run size 729, the largest group size reported in the existing literature is 32 \cite{li2025grouped}. Construction~\ref{construction1} can generate a grouped orthogonal array  with two groups of size 40, whereas Construction~\ref{construction6} generates a grouped orthogonal array with three groups of size 36.
\end{example}

\renewcommand\arraystretch{0.7}
\setlength{\tabcolsep}{3.25pt}
\begin{table}[htpb]
\setlength{\abovecaptionskip}{0.cm}
\setlength{\belowcaptionskip}{0.2cm}
\centering  
\small
\caption{\small {Grouped orthogonal arrays from Construction \ref{construction6}}.\label{table77}}
\begin{tabular}{lll}
\hline
\textbf{A} & \textbf{B} & {Constructed grouped orthogonal array}  \\
\hline
$GOA(27, 4\times 3, 3\times 3, 3, 2)$& $D_3(18, 5, 3)$& $GOA(486, 20\times 3, 3\times 3, 3, 2)$\\
& $D_3(27, 9, 3)$& $GOA(729, 36\times 3, 3\times 3, 3, 2)$\\

 {$GOA(81, 10\times 4, 3\times 4, 3, 2)$}& $D_3(18, 5, 3)$&{ $GOA(1458, 50\times 4, 3\times 4, 3, 2)$}\\
& $D_3(27, 9, 3)$& { $GOA(2187, 90\times 4, 3\times 4, 3, 2)$}\\

$GOA(162, (8\times 6, 4\times 2), 3\times 8, 3, 2)$& $D_3(9, 4, 3)$& $GOA(1458, (32\times 6, 16\times 2), 3\times 8, 3, 2)$\\
& $D_3(18, 5, 3)$& $GOA(2916, (40\times 6, 20\times 2), 3\times 8, 3, 2)$\\
 
 
$GOA(64, 4\times 5, 3\times 5, 4, 2)$& $D_3(16, 6, 4)$& $GOA(1024, 24\times 5, 3\times 5, 4, 2)$\\
& $D_3(64, 16, 4)$& $GOA(4096, 64\times 5, 3\times 5, 4, 2)$\\

$GOA(125, 4\times 6, 3\times 6, 5, 2)$ & $D_3(25, 6, 5)$& $GOA(3125, 24\times 6, 3\times 6, 5, 2)$\\
 \hline
\end{tabular}
\end{table}

Table~\ref{table77} summarizes some grouped orthogonal arrays constructed via Construction~\ref{construction6}.  For $GOA(486, 20\times 3, 3\times 3, 3, 2)$, a grouped orthogonal array with the same design parameters can be constructed via \cite{chen2025}. All other grouped orthogonal arrays are new.

\renewcommand\arraystretch{0.8}
\setlength{\tabcolsep}{4.5pt}
\begin{table}[htpb]
\scalebox{0.93}{
\setlength{\abovecaptionskip}{0.cm}
\setlength{\belowcaptionskip}{0.2cm}
\centering \scriptsize
\begin{threeparttable}
\caption{\small {Comparisons with \cite{chen2025} and \cite{li2025grouped}.\label{table88}}}
\begin{tabular}{llll}
\hline
$N$&The proposed GOA& \cite{chen2025}&\cite{li2025grouped}\\
\hline
162
&$GOA(162, (8 \times 6, 4\times 2,4), (3 \times 8,2), 3, 2)^{C_2}$&{ $GOA(162, 8\times 3, 3 \times 3, 3, 2)^\ddagger$}&$GOA(162, 8 \times 6, 3 \times 6, 3, 2)$\\
&$GOA(162, 4\times 3, 4\times 3, 3, 2)^{C_5}$&&\\
&&&\\
243
& $GOA(243, (8 \times 12, 4\times 2,4), (3 \times 14,2), 3, 2)^{C_2}$& { $GOA(243, 6 \times 20, 5 \times 20, 3, 2)$}&$GOA(243,16\times 4, 3\times 4,3,2)$\\
&$GOA(243, (16 \times 6,4\times 3), 3 \times 9, 3, 2)^{C_3}$&{ $GOA(243, 7 \times 17, 4 \times 17, 3, 2)$}&\\
&$GOA(243, 5\times 4, 4\times 4, 3, 2)^{C_4}$&{ $GOA(243, 8 \times 15, 4\times 15, 3, 2)$}&\\
&$GOA(243, 4\times 3, 4\times 3, 3, 3)^{C_5}$&{ $GOA(243, 9 \times 13, 3 \times 13, 3, 2)$}&\\

&&&\\
486
&   
 $GOA(486, (8 \times 24, 4\times 2,4), (3 \times 26,2), 3, 2)^{C_2}$
 &{ GOA$(486,8\times 9, 3\times 9,3,2)^\ddagger$}&$GOA(486, 16\times 6, 3 \times 6, 3, 2)$\\
&$GOA(486, (16 \times 9, 8\times 3, 4\times 3), 3 \times 15, 3, 2)^{C_3}$&{ GOA$(486,20\times 3, 3\times 3,3,2)^\ddagger$}&\\
&$GOA(486, (20 \times 2, 5), 3 \times 3, 3, 2)^{C_1}$&&\\
&{ $GOA(486, 20\times 3, 3\times 3, 3, 2)^{C_6}$}&&\\

&&&\\
729
 & $GOA(729, (8 \times 40,40), (3 \times 40,2), 3, 2)^{C_2}$&{ $GOA(729, 7 \times 52, 6 \times 52, 3, 2)$}&$GOA(729, 32 \times 8, 3 \times 8, 3, 2)$\\
 &$GOA(729, (16 \times 18, 8\times 3, 4\times 3), 3 \times 24, 3, 2)^{C_3}$&{ $GOA(729, 8 \times 45, 5 \times 45, 3, 2)$}&\\
 &$GOA(729, (32 \times 8, 8\times 4), 3 \times 12, 3, 2)^{C_3}$&{ $GOA(729, 9 \times 40, 4 \times 40, 3, 2)$}&\\
 &$GOA(729, 36\times 3, 3\times 3, 3, 2)^{C_6}$&{ $GOA(729, 10 \times 36, 4 \times 36, 3, 2)$}&\\
 &$GOA(729, (40 \times 2, 10), 3 \times 3, 3, 2)^{C_1}$&&\\
 &$GOA(729, 9\times 4, 4\times 4, 3, 2)^{C_4}$&&\\&$GOA(729, 5\times 13, 4\times 13, 3, 2)^{C_4}$&&\\
&$GOA(729, 5\times 4, 4\times 4, 3, 3)^{C_4}$&&\\
&$GOA(729, 6\times 3, 4\times 3, 3, 3)^{C_5}$&&\\

&&&\\
625&$GOA(625, (8 \times 18,6), (3 \times 18,2), 5, 2)^{C_2}$&{ $GOA(625, 5 \times 31, 4 \times 31, 5, 2)$}&$GOA(625,8\times 12, 3\times 12, 5,2)$\\
&&{ $GOA(625, 6 \times 26, 3 \times 26, 5, 2)$}&\\
&&{ $GOA(625, 7 \times 22, 3 \times 22, 5, 2)$}&\\
&&{ $GOA(625, 8 \times 19, 3 \times 19, 5, 2)$}&\\
     &&&\\

 3125& GOA$(3125,(8\times 90,4\times 6,31),(3\times 96,2),5,2)^{C_2}$ & { $GOA(3125, (12 \times 12, 6), 3 \times 13, 5, 2)^\ddagger$}&$GOA(3125,16\times 24, 3\times 24,5,2)$\\
 &$GOA(3125, (16 \times 36, 4\times 6), 3 \times 42, 5, 2)^{C_3}$&&\\
 &$GOA(3125, (24 \times 6, 6), 3 \times 7, 5, 2)^{C_1}$&&\\
 &$GOA(3125, 7\times 4, 4\times 4, 5, 2)^{C_5}$&&\\

 \hline
\end{tabular}
\begin{tablenotes}
\scriptsize{
\item {$C_1$–$C_6$ indicate that the corresponding grouped orthogonal arrays are obtained from Constructions~\ref{construction1}–\ref{construction6}, respectively.
$\ddagger$ indicates those designs may be either regular or nonregular. Other arrays in the column of \cite{chen2025} are regular.}}
\end{tablenotes}
\end{threeparttable}}
\end{table}

 \section{Conclusions and Discussion}

We propose six construction methods for grouped orthogonal arrays. Constructions~\ref{construction1}--\ref{construction3} generate grouped orthogonal arrays by combining two orthogonal arrays, resulting in designs with within-group strength three.  Constructions~\ref{construction4}--\ref{construction6} employ difference schemes and orthogonal arrays of higher strength, enabling the construction of grouped orthogonal arrays with higher within-group and/or overall strength. Overall, the proposed methods substantially expand the class of available grouped orthogonal arrays by allowing more flexible run sizes and group sizes, and/or higher attainable strengths than existing approaches. For simplicity of presentation and proof, we use the Kronecker product in Constructions~\ref{construction4}--\ref{construction6}; however, the generalized Kronecker product can also be applied. Collectively, the proposed constructions provide a large pool of grouped orthogonal arrays, from which appropriate optimal designs for analysis can be selected using design optimality criteria.

Table~\ref{table88} compares our constructions with those in \cite{chen2025} and \cite{li2025grouped}, and these results demonstrate the effectiveness, flexibility, and broad applicability of our proposed construction methods. We highlight three key advantages of our approaches.

First, for grouped designs with within-group strength three and overall strength two, our method can produce more groups under the same run size and grouping structure. For example, given a run size of 486 and a group size of 16, the $GOA(486, (16 \times 9, 8\times 3, 4\times 3), 3 \times 15, 3, 2)$ from Construction \ref{construction3} can accommodate 9 groups of size 16, whereas the $GOA(486, 16\times 6, 3 \times 6, 3, 2)$ from \cite{li2025grouped} can accommodate only 6 such groups. Similarly, for a run size of 486 and a group size of 8, the  $GOA(486, (8 \times 24, 4\times 2,4), (3 \times 26,2), 3, 2)$ from Construction \ref{construction2} yields 24 groups of size 8, while the $GOA(486, 8 \times 9, 3 \times 9, 3, 2)$ from \cite{chen2025} produces only 9 such groups.

Second, our methods can construct grouped orthogonal arrays with larger group sizes.  For instance, for a run size of 3125, Construction~\ref{construction1} yields a grouped orthogonal array $GOA(3125, (24 \times 6, 6), 3 \times 7, 5, 2)$ with the group size 24, whereas the largest group sizes reported in 
\cite{chen2025} and \cite{li2025grouped}  are 12 and 16, respectively.

Third, our construction methods can generate grouped orthogonal arrays with higher within-group and/or overall strength. For example, Construction~\ref{construction5} yields the arrays $GOA(243,$ $ 4\times 3, 4\times 3, 3, 3)$ and $GOA(729, 6\times 3, 4\times 3, 3, 3)$, both of which have within-group strength four and overall strength three. Grouped orthogonal arrays with these structures cannot be obtained using existing methods. 

Several directions merit further investigation. The present work focuses on the construction of grouped orthogonal arrays without explicitly addressing their statistical optimality. Future research may examine the optimality properties of these designs under criteria such as generalized minimum aberration  and minimum moment criteria \cite{deng1999generalized, wu2001generalized,xu2003minimum}. The proposed constructions may be extended to mixed-level cases \cite{pang2021construction, wang2025small}. Developing systematic construction methods for mixed-level grouped orthogonal arrays would substantially broaden the scope of applicability. In this context, it may be worthwhile to investigate the impact of level asymmetry on aliasing structures and projection properties. Another direction is to extend grouped orthogonal arrays to incorporate different data structures involving branching factors and nested factors \cite{zhang2026bayesian}.

\section*{Appendix}
 
To prove the propositions, we first state the following lemmas and definitions. Definition 1 below defines an operator between two matrices $\bA$ and $\bB$ that have the same number of columns. The resulting matrix has the same number of columns as 
$\bA$ and $\bB$, but its number of rows is equal to the product of the number of rows of $\bA$ and that of $\bB$. 

Definition \ref{defn1} was originally introduced in \cite{mukhopadhyay1981construction}. While the original paper used the symbol $\oplus$ for the operation, we use $\biguplus$ in this paper to distinguish it from the Kronecker sum notation.
\begin{definition}\label{defn1}
Let $\boldsymbol{A}$ be an $n_1 \times m$ array and $\boldsymbol{B}$ be an $n_2 \times m$ array, both with entries from $\mathrm{GF}(s)$. The operation $\boldsymbol{A} \biguplus \boldsymbol{B}$ constructs an $(n_1 n_2) \times m$ matrix formed by the row-wise sum of $\boldsymbol{A}$ and $\boldsymbol{B}$. Specifically, for each row $\boldsymbol{u}_i$ of $\boldsymbol{A}$ ($i = 1, \dots, n_1$) and each row $\boldsymbol{w}_j$ of $\boldsymbol{B}$ ($j = 1, \dots, n_2$), the sum $\boldsymbol{u}_i + \boldsymbol{w}_j$ becomes a row of $\boldsymbol{A} \biguplus \boldsymbol{B}$.
\end{definition}

Let $\ba$, $\bc$ and $\be$ be a column vector of length $n_1$, and $\bb$, $\bd$  and $\bff$ be a column vector of length $n_2$, whose entries are from $GF(s)$. Lemmas~\ref{lem:2} and \ref{lem:3} provide the conditions 
an array constructed by Kronecker product to be  of strength 2 and strength 3, respectively. Both lemmas were given by \cite{2022He}. 

\begin{Lemma}\label{lem:2}
Given $\ba, \bb, \bc, \bd$ defined above, the array $({\ba \oplus \bb, \bc \oplus \bd})$ is an $OA(n_1n_2, 2, s, 2)$ if one of the following conditions is satisfied, and $\alpha\in GF(s)$,
\begin{itemize}
  \item[(i)]   $(\bb, \bd)$ is an $OA(n_2, 2, s, 2)$;
  \item[(ii)] $\bb$ is balanced, $\bd = \alpha\bb$, and   one of the following four conditions is satisfied:
                 $(1)$ $\alpha=1$ and $(\ba, \bc)$ is a $D(n_1, 2, s)$;
                 $(2)$ $(\ba,\bc)$ is an $OA(n_1,2, s, 2)$; $(3)$ $\bc=\bzero$, $\alpha\neq 0$, and $\ba$ is balanced;
                 and $(4)$ $\ba=\bc$, $\alpha\neq 1$, and $\ba$ is balanced; and
  \item[(iii)] $\bd = \bzero$, and one of the following two conditions is satisfied: $(1)$ $\bb$ and $\bc$ are
                 balanced; and $(2)$ $\bb=\bzero$ and $(\ba, \bc)$ is an $OA(n_1, 2, s, 2)$.
\end{itemize}
\end{Lemma}

\begin{Lemma}\label{lem:3}
Given $\ba, \bb, \bc, \bd, \be, \bff$ defined above, the array ($\ba\oplus\bb$, $\bc\oplus\bd$, $\be\oplus\bff$) is an $OA(n_1n_2, 3, s, 3)$ if one of the following conditions is satisfied, and $\alpha, \beta\in GF(s)$,
\begin{itemize}
  \item[(i)]  $(\bb, \bd,\bff)$ is an $OA(n_2, 3, s, 3)$;
  \item[(ii)]  $\bff= \alpha\bb$, $(\bb,\bd)=OA(n_2, 2, s, 2)$,
                 and $(\be, \alpha\ba)=D(n_1, 2, s)$; 
  \item[(iii)]  $ \bd=  \beta\bb$, $\bff= \alpha\bb$, $\bb$ is balanced, $(\ba, \bc)=OA(n_1, 2, s, 2)$
                 and one of the following three conditions is satisfied:
                 $(1)$ $\be=\ba$ and $\alpha\neq 1$; or $(2)$ $\be=\bzero$ and $\alpha\neq 0$; and (3) $(\ba,\bc,\be) =OA(n_1,3,s,3)$; and
  \item[(iv)] $\bb = \bd = \bzero$, $\bff$ is balanced, and $(\ba, \bc) = OA(n_1,2,s,2)$.
\end{itemize}
\end{Lemma}

The following lemmas serve as important construction methods for orthogonal arrays. Lemma \ref{lem5} is due to Lemma 6.27 of \cite{hedayat1999orthogonal} and provides a method for constructing an orthogonal array of strength two from a difference scheme  and an orthogonal array both of strength two. Lemmas   \ref{lemt+1},
\ref{lem7}, \ref{lemDt}, and \ref{conststr3}  are Lemmas 4.1, 4.2, 4.3 and Theorem 4.1.2 in \cite{mukhopadhyay1981construction}. 

\begin{Lemma}\label{lem5}
    Suppose $\boldsymbol{A}$ is an $OA(n_1, m_2, s, 2)$ and $\boldsymbol{D}$ is a $D(n_2, m_2, s)$, 
    both based on the $GF(s)$, then the array $\boldsymbol{A}\oplus \boldsymbol{D}$ is an  $OA(n_1n_2, m_1m_2, s, 2)$.
\end{Lemma}
\begin{Lemma} \label{lemDt}
    Suppose $\boldsymbol{B}$ is a $D_t(r, c, s)$ and a vector $\boldsymbol{v}=(\alpha_1, \ldots, \alpha_m)^\top$ with $m=qs$ for some integer $q$ such that each element of $GF(s)$ occur $q$ times as an entry in $\boldsymbol{v}$, then $\boldsymbol{v}\oplus \boldsymbol{B}$ is an $OA(qsr, c, s, t)$.
\end{Lemma}

\begin{Lemma}\label{lem7}
    If $\boldsymbol{A}$ is an $OA(n, t, s, t)$ and $\boldsymbol{v} = (\alpha_1, \alpha_2, \ldots, \alpha_t)$ is any $t$-tuple on $GF(s)$, then $D = \boldsymbol{v} \biguplus \boldsymbol{A}$ is also an $OA(n,t,s,t)$.
\end{Lemma}

\begin{Lemma}\label{lemt+1}
    Let $\boldsymbol{A}$ be an $n \times t$ array, where the first $t-1$ columns form an $OA(n, t-1, s, t-1)$, and the $t$th column is identical to the $(t-1)$th column. Let $\boldsymbol{B}$ be an $m \times t$ array in which the last two columns constitute a $D(m, 2, s)$. Then, $\boldsymbol{A} \biguplus \boldsymbol{B}$ is an $OA(nm, t, s, t)$.
\end{Lemma}

\begin{Lemma}\label{conststr3}
Let \( \boldsymbol{A} = OA(N_1, r_1, s, 3) \) and \( \boldsymbol{B} = D_3(N_2, r_2, s) \). Then, \(\boldsymbol{A} \oplus\boldsymbol{B}\) is an   $OA(N_1N_2,$ $ r_1 r_2, s, 3)$.
\end{Lemma}

\subsection*{Proof of Proposition \ref{proposition1}}
\begin{proof}
The design $\textbf{D}$ is an orthogonal array of strength two following by Lemma~\ref{lem:2}. 
To show $\textbf{E}_k$ is an 
$OA(n_1n_2,4m_2,s,3)$ or an $OA(n_1n_2,2m_2,s,3)$ or an 
$OA(n_1n_2,m_2,s,3)$, we examine every three columns of $\textbf{E}_k$. Without loss of generality, nine possible cases are considered for generating the three columns of $\textbf{E}_k$. They are, (1) $(\ba_{i_1} \oplus \bb_{j_1}, \ba_{i_1} \oplus \bb_{j_2}, \ba_{i_1} \oplus \bb_{j_3})$; (2) $(\ba_{i_1} \oplus \bb_{j_1}, \ba_{i_1} \oplus \bb_{j_2}, \ba_{i_1} \oplus \alpha \bb_{j_1})$; (3) $(\ba_{i_1} \oplus \bb_{j_1}, \ba_{i_2} \oplus \bb_{j_1}, \ba_{i_1} \oplus \bb_{j_2})$; (4) $(\ba_{i_1} \oplus \bb_{j_1}, \ba_{i_2} \oplus \bb_{j_1}, \ba_{i_1} \oplus \alpha \bb_{j_2})$; (5) $(\ba_{i_1} \oplus \bb_{j_1}, \ba_{i_2} \oplus \bb_{j_2}, \ba_{i_1} \oplus \bb_{j_2})$; (6) $(\ba_{i_1} \oplus \bb_{j_1}, \ba_{i_2} \oplus \bb_{j_1}, \ba_{i_3} \oplus \bb_{j_1})$;
 (7) $(\ba_{i_1} \oplus \bb_{j_1}, \ba_{i_2} \oplus \bb_{j_1}, \ba_{i_3} \oplus \bb_{j_2})$; (8) $(\ba_{i_1} \oplus \bb_{j_1}, \ba_{i_2} \oplus \bb_{j_1}, \ba_{i_3} \oplus \alpha \bb_{j_1})$;
 (9) $(\ba_{i_1} \oplus \bb_{j_1}, \ba_{i_2} \oplus \bb_{j_1}, \ba_{i_3} \oplus \alpha \bb_{j_2})$.  Here $\ba_{i_1}$, $\ba_{i_2}$ and $\ba_{i_3}$ are three columns in 
$\bA^{(i_1)},$
$\bA^{(i_2)}$ and
$\bA^{(i_3)},$
 respectively,  while
 $\bb_{j_1},  \bb_{j_2}$ and $ \bb_{j_3}$ are the $j_1$th, $j_2$th, and $j_3$th columns of $\textbf{B}$ for $i_1 \neq i_2 \neq i_3$, and $j_1 \neq j_2 \neq j_3$, and $\alpha \in GF(s)$  \textbackslash  \{1\}. All these nine cases lead to an orthogonal array of strength three followed by Lemma~\ref{lem:3}. 
\end{proof}

\subsection*{Proof of Proposition \ref{proposition2}}

\begin{proof} 
Proposition \ref{proposition2} follows directly from the proof of Proposition \ref{proposition1} with $\bB$ being a strength-two orthogonal array with two columns, and hence the proof is omitted.
\end{proof}

\subsection*{Proof of Proposition \ref{proposition3}}
\begin{proof} 
Proposition \ref{proposition3} follows directly from the proof Proposition \ref{proposition1} with $\bB$ being a strength-three orthogonal array with more than three columns, and hence the proof is omitted.
\end{proof}

 \subsection*{Proof of Proposition \ref{proposition4}}
\begin{proof}
 By Lemma~\ref{conststr3}, when $t_1 = 3$ and $t_2 \geq 3$, the array $\boldsymbol{A} \oplus \boldsymbol{B}$ forms an orthogonal array of strength three.
    When $t_1 = 2$ and $t_2 \geq 3$, $\boldsymbol{A} \oplus \boldsymbol{B}$ becomes an orthogonal array of strength two, as proven in Lemma~\ref{lem5}. 
Furthermore, each group  $\boldsymbol{E}_i = \boldsymbol{a}_i \oplus \boldsymbol{B}$ of  $\boldsymbol{E}$ constitutes an orthogonal array of strength $t_2$ according to Lemma~\ref{lemDt}.
\end{proof}

 \subsection*{Proof of Proposition \ref{proposition5}}

\begin{proof}
   From Construction~\ref{construction5}, each group of \(\boldsymbol{E}\) is given by
\begin{align}\label{p3.6_1}
    \boldsymbol{E}_{i} = [\boldsymbol{A} \oplus \boldsymbol{b}_{i},\; \boldsymbol{a}_{i} \oplus \boldsymbol{b}_{m_2}]
\end{align}
for \(1 \leq i \leq \min(m_1, m_2 - 1)\). We claim that each \(\boldsymbol{E}_{i}\) is an orthogonal array of strength \(t_1\), and in the special case when \(m_1 = t_1\), \(\boldsymbol{E}_{i}\) is an orthogonal array of strength \(t_1 + 1\).

To prove these claims, we evaluate  \eqref{p3.6_1} as follows,
    \begin{align*}
         \boldsymbol{E}_{i} &= \left[\boldsymbol{A} \oplus \boldsymbol{b}_{i}, \boldsymbol{a}_{i} \oplus \boldsymbol{b}_{m_2}\right] 
        = \left[\boldsymbol{a}_1 \oplus \boldsymbol{b}_{i}, \ldots, \boldsymbol{a}_{m_1} \oplus \boldsymbol{b}_{i}, \boldsymbol{a}_{i} \oplus \boldsymbol{b}_{m_2}\right] \\
        &= \left[\boldsymbol{a}_1, \ldots, \boldsymbol{a}_{m_1}, \boldsymbol{a}_{i}\right] \biguplus \left[\boldsymbol{b}_{i}, \ldots, \boldsymbol{b}_{i}, \boldsymbol{b}_{m_2}\right].
    \end{align*}
       
    Let \( N = n_1 n_2 \), and let \(\boldsymbol{E}_{\text{sub}}\) denote an arbitrary \( N \times t_1 \) subarray of \(\boldsymbol{E}_{i}\). To show that \(\boldsymbol{E}_i\) is an orthogonal array of strength \( t_1 \), it suffices to show that any such \(\boldsymbol{E}_{\text{sub}}\) has strength \( t_1 \). We consider the following two cases:

    \textbf{Case I}:
    All columns of \( \boldsymbol{E}_{\text{sub}} \) are among the first \( m_1 \) columns of \(  \boldsymbol{E}_{i} \).  
    In this case,  
    \begin{align*}
        \boldsymbol{E}_{\text{sub}} = [\boldsymbol{a}_{j_1}, \ldots, \boldsymbol{a}_{j_{t_1}}] \biguplus [\boldsymbol{b}_{i}, \ldots, \boldsymbol{b}_{i}] 
        = [\boldsymbol{a}_{j_1}, \ldots, \boldsymbol{a}_{j_{t_1}}] \oplus \boldsymbol{b}_{i}
    \end{align*}  
    for \( 1 \leq j_1 < \cdots < j_{t_1} \leq m_1 \). Since \( [\boldsymbol{a}_{j_1}, \ldots, \boldsymbol{a}_{j_{t_1}}] \) consists of \( t_1 \) columns from \( \boldsymbol{A} \), where \( \boldsymbol{A} \) is an \( OA(n_1, m_1, s, t_1) \), the submatrix \( [\boldsymbol{a}_{j_1}, \ldots, \boldsymbol{a}_{j_{t_1}}] \) forms an \( OA(n_1, t_1, s, t_1) \), i.e., an orthogonal array of strength \( t_1 \). It follows that \( \boldsymbol{E}_{\text{sub}} \), which consists of replicates of this \( OA(n_1, t_1, s, t_1) \), is itself an orthogonal array of strength \( t_1 \).

\textbf{Case II}:  \( \boldsymbol{E}_{\text{sub}} \) contains the last column of \( \boldsymbol{E}_{i} \). 
    Here,  
    \begin{align*}
        { \boldsymbol{E}_{\text{sub}}} = [\boldsymbol{a}_{j_1}, \ldots, \boldsymbol{a}_{j_{t_1-1}}, \boldsymbol{a}_{i}] \biguplus [\boldsymbol{b}_{i}, \ldots, \boldsymbol{b}_{i}, \boldsymbol{b}_{m_2}]
    \end{align*}
    for \( 1 \leq j_1 < \cdots < j_{t_1-1} \leq k \) and \( 1 \leq i \leq \min(m_1, m_2-1) \).  

    If \( i \in \{j_1, \ldots, j_{t_1-1}\} \), then by Lemma \ref{lemt+1}, { \( \boldsymbol{E}_{\text{sub}} \)} is of strength \( t_1 \).  
   If \( i \notin \{ j_1, \ldots, j_{t_1 - 1} \} \), then the submatrix \( [\boldsymbol{a}_{j_1}, \ldots, \boldsymbol{a}_{j_{t_1 - 1}}, \boldsymbol{a}_{i}] \) consists of \( t_1 \) columns of \( \boldsymbol{A} \), where \( \boldsymbol{A} \) is an \( OA(n_1, m_1, s, t_1) \). Therefore, \( [\boldsymbol{a}_{j_1}, \ldots, \boldsymbol{a}_{j_{t_1 - 1}}, \boldsymbol{a}_{i}] \) itself is an \( OA(n_1, t_1, s, t_1) \), i.e., an orthogonal array of strength \( t_1 \). By Lemma~\ref{lem7}, it follows that { \( \boldsymbol{E}_{\text{sub}} \)} is also of strength \( t_1 \).

    Finally, for the special case when \( m_1 = t_1 \), an arbitrary \( N \times (t_1+1) \) subarray of { \( \boldsymbol{E}_{i} \)} is given by  
    \begin{align*}
        [\boldsymbol{a}_1, \ldots, \boldsymbol{a}_{t_1}, \boldsymbol{a}_{i}] \biguplus [\boldsymbol{b}_{i}, \ldots, \boldsymbol{b}_{i}, \boldsymbol{b}_{m_2}].
    \end{align*} 
    This satisfies the conditions stated in Lemma~\ref{lemt+1}, which guarantees that the resulting design is an orthogonal array of strength \( t_1 + 1 \).

{ Note that $t_1\geq 3.$} By Lemma~\ref{conststr3}, if { \( t_2=3 \)}, then \( \boldsymbol{A} \oplus \boldsymbol{B} \) is an orthogonal array of strength 3. By Lemma~\ref{lem5}, if   { \( t_2 = 2 \),} then \( \boldsymbol{A} \oplus \boldsymbol{B} \) is an orthogonal array of strength 2.
{ 
All columns of $\boldsymbol{E}$ are drawn from those of \( \boldsymbol{A} \oplus \boldsymbol{B} \), and thus $\boldsymbol{E}$ is also an orthogonal array of strength 2.
This completes the proof.}
\end{proof}

\subsection*{Proof of Proposition \ref{proposition6}}
\begin{proof}
From Lemma~\ref{conststr3}, $\boldsymbol{A}_i$ is an orthogonal array of strength 3 and $\boldsymbol{B}$ is a difference scheme of strength 3 implies that $\boldsymbol{A}_i \oplus \boldsymbol{B}$ is an orthogonal array of strength 3. Furthermore, by Lemma~\ref{lem5}, the array $A \oplus B$ is an orthogonal array of strength 2.
\end{proof}

\par

\section*{Acknowledgments}

Wang was supported by the National Natural Science Foundation of China (Grant No.
12301323). Liu and Lin were supported by a Discovery grant from the Natural Sciences and Engineering Research Council of Canada. Chen was supported by the National Natural Science Foundation of China (Grant No. 12401325).
\par


\bibliographystyle{chicago}   
\bibliography{newbibfile} 

\end{document}